\documentclass[runningheads]{llncs}
\usepackage{amsmath}
\usepackage{amssymb}

\usepackage{cite}

\DeclareRobustCommand{\VAN}[3]{#2} 

\usepackage{color}
\usepackage{tikz}
\usetikzlibrary{positioning}
\usetikzlibrary{shapes}
\usetikzlibrary{patterns}
\usetikzlibrary{arrows}
\usetikzlibrary{decorations.pathmorphing}
\tikzset{every picture/.style={>=stealth'}}

\usepackage{graphicx}
\usepackage{pgfplots}
\pgfplotsset{compat=1.7}
\usetikzlibrary{plotmarks}

\makeatletter
\newenvironment{customlegend}[1][]{%
    \begingroup
    \pgfplots@init@cleared@structures
    \pgfplotsset{#1}%
}{%
    \pgfplots@createlegend
    \endgroup
}%

\def\addlegendimage{\pgfplots@addlegendimage}
\makeatother

\colorlet{adiar-skip_transpose}{green!50!cyan!80!black}
\colorlet{adiar-pruning_and}{cyan!80!black}
\colorlet{adiar-exists_replace}{blue!40!purple}
\colorlet{adiar-shift_replace}{red}

\tikzstyle{dots_adiar-naive}=[solid, color=black]
\tikzstyle{dots_adiar-skip_transpose}=[only marks, mark=diamond, mark size=2.2pt, mark options={color=adiar-skip_transpose, fill=adiar-skip_transpose, line width=1pt, opacity=0.5}]
\tikzstyle{dots_adiar-pruning_and}=[only marks, mark=halfdiamond*, mark size=2.2pt, mark options={color=adiar-pruning_and, fill=adiar-pruning_and, line width=1pt, opacity=0.5}]
\tikzstyle{dots_adiar-exists_replace}=[only marks, mark=halfdiamond*, mark size=2.2pt, mark options={rotate=180, color=adiar-exists_replace, fill=adiar-exists_replace, line width=1pt, opacity=0.5}]
\tikzstyle{dots_adiar-shift_replace}=[only marks, mark=diamond*, mark size=2.2pt, mark options={color=adiar-shift_replace, fill=adiar-shift_replace, line width=1pt, opacity=0.5}]

\tikzstyle{x_adiar-skip_transpose}=[only marks, mark=+, mark size=2.2pt, mark options={color=adiar-skip_transpose, fill=adiar-skip_transpose, opacity=0.7}]
\tikzstyle{x_adiar-pruning_and}=[only marks, mark=+, mark size=2.2pt, mark options={color=adiar-pruning_and, fill=adiar-pruning_and, opacity=0.7}]
\tikzstyle{x_adiar-exists_replace}=[only marks, mark=+, mark size=2.2pt, mark options={color=adiar-exists_replace, fill=adiar-exists_replace, opacity=0.7}]
\tikzstyle{x_adiar-shift_replace}=[only marks, mark=+, mark size=2.2pt, mark options={color=adiar-shift_replace, fill=adiar-shift_replace, opacity=0.7}]

\colorlet{buddy}{blue!40!purple}
\colorlet{cal}{blue!80!cyan}
\colorlet{cudd}{orange}
\colorlet{libbdd}{green!60!black}
\colorlet{sylvan}{cyan}

\tikzstyle{dots_buddy}=[only marks, mark=pentagon*, mark size=2.2pt, mark options={color=buddy, fill=buddy, opacity=0.6}]
\tikzstyle{dots_cal}=[only marks, mark=triangle*, mark size=2.2pt, mark options={rotate=180, color=cal, fill=cal, opacity=0.6}]
\tikzstyle{dots_cudd}=[only marks, mark=triangle*, mark size=2.2pt, mark options={color=cudd, fill=cudd, opacity=0.6}]
\tikzstyle{dots_libbdd}=[only marks, mark=square*, mark size=2.2pt, mark options={color=libbdd, fill=libbdd, opacity=0.6}]
\tikzstyle{dots_sylvan}=[only marks, mark=*, mark size=2.2pt, mark options={color=sylvan, fill=sylvan, opacity=0.6}]

\tikzstyle{plot_adiar}=[color=adiar-shift_replace, line width=0.7pt, mark=diamond*, mark size=2.2pt, mark options={color=adiar-shift_replace, fill=adiar-shift_replace, line width=1pt, opacity=0.5}]
\tikzstyle{plot_buddy}=[color=buddy, line width=0.7pt, mark=pentagon*, mark size=2.2pt, mark options={color=buddy, fill=buddy, opacity=0.6}]

\tikzstyle{x_buddy}=[only marks, mark=asterisk, mark size=2.3pt, mark options={color=buddy, fill=buddy, opacity=1}]

\newcommand{\markSkipTranspose}[0]{%
  \protect\tikz[x=1.2ex,y=1.85ex,line width=.1ex,line join=round, baseline=1.7pt]{%
    \draw[color=adiar-skip_transpose, line width=1pt, opacity=0.6]%
      (0,0.5) -- (0.5,1) -- (1,0.5) -- (0.5,0) -- (0,0.5) -- cycle;%
  }%
}

\newcommand{\markPruningAnd}[0]{%
  \protect\tikz[x=1.2ex,y=1.85ex,line width=.1ex,line join=round, baseline=1.7pt]{%
    \fill[color=adiar-pruning_and, opacity=0.6]%
      (1,0.5) -- (0.5,0) -- (0,0.5) -- cycle;%
    \draw[color=adiar-pruning_and, line width=1pt=1pt, opacity=0.6]%
      (0,0.5) -- (0.5,1) -- (1,0.5) -- (0.5,0) -- (0,0.5) -- cycle;%
  }%
}

\newcommand{\markExistsReplace}[0]{%
  \protect\tikz[x=1.2ex,y=1.85ex,line width=.1ex,line join=round, baseline=1.7pt]{%
    \fill[color=adiar-exists_replace, opacity=0.6]%
      (0,0.5) -- (0.5,1) -- (1,0.5) -- cycle;%
    \draw[color=adiar-exists_replace, line width=1pt, opacity=0.6]%
      (0,0.5) -- (0.5,1) -- (1,0.5) -- (0.5,0) -- (0,0.5) -- cycle;%
  }%
}

\newcommand{\markShiftReplace}[0]{%
  \protect\tikz[x=1.2ex,y=1.85ex,line width=.1ex,line join=round, baseline=1.7pt]{%
    \draw[color=adiar-shift_replace, fill=adiar-shift_replace, line width=1pt, opacity=0.6]%
      (0,0.5) -- (0.5,1) -- (1,0.5) -- (0.5,0) -- (0,0.5) -- cycle;%
  }%
}

\usepackage{hyperref,doi}
\usepackage[capitalise]{cleveref}
\usepackage{caption,subcaption}
\usepackage{multirow}
\hypersetup{
    colorlinks,
    linkcolor={red!50!black},
    citecolor={green!40!black},
    urlcolor={blue!40!black}
}
\usepackage{array}
\newcolumntype{H}{>{\setbox0=\hbox\bgroup}c<{\egroup}@{}}

\usepackage{graphicx}
\newcommand{\N}[0]{\ensuremath{\mathbb{N}}}

\newcommand{\Oh}[1]{\ensuremath{\mathcal{O} ( #1 )}}

\newcommand{\sort}[1]{\text{sort} ( #1 )}

\newcommand{\arc}[3][solid]{
  \ensuremath{#2}
  \,
  \tikz[baseline=-\the\dimexpr\fontdimen22\textfont2\relax]{
    \draw[->, #1](0,0) -- ++(1.2em,0);
  }
  \,
  \ensuremath{#3}
}

\def\orcidID#1{\smash{\href{http://orcid.org/#1}{\protect\raisebox{-1.25pt}{\protect\includegraphics{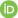}}}}}

\spnewtheorem{optimisation}{Optimisation}{\bfseries}{\itshape}
\spnewtheorem*{recommendation}{Recommendation}{\bfseries}{\itshape}

\newcommand*{\mailto}[1]{\href{mailto:#1}{\nolinkurl{#1}}}

\usepackage{bbding}

\def\arxiv{1}

\title{Symbolic Model Checking\\in External Memory}

\titlerunning{Relational Product of BDDs in External Memory}

\author
{
  Steffan Christ S{\o}lvsten
  \orcidID{0000-0003-0963-6569}
  \and
  Jaco van de Pol
  \orcidID{0000-0003-4305-0625}
}

\authorrunning
{
  S. C. S{\o}lvsten and J. van de Pol
}

\institute
{
  Aarhus University, Denmark
  \texttt{\{%
    \href{mailto:soelvsten@cs.au.dk}{\color{black} soelvsten},%
    \href{mailto:jaco@cs.au.dk}{\color{black} jaco}%
    \}@cs.au.dk}
}

\begin{document}

\maketitle


\begin{abstract}
  We extend the external memory BDD package Adiar with support for monotone variable substitution.
  Doing so, it now supports the relational product operation at the heart of symbolic model
  checking. We also identify additional avenues for merging variable substitution fully and the
  conjunction operation partially inside the relational product's existential quantification step.
  For smaller BDDs, these additional ideas improve the running of Adiar for model checking tasks up
  to $47\%$. For larger instances, the computation time is mostly unaffected as it is dominated by
  the existential quantification.

  Adiar's relational product is about one order of magnitude slower than conventional depth-first
  BDD implementations. Yet, its I/O-efficiency allows its running time to be virtually independent
  of the amount of internal memory. This allows it to compute on BDDs with much less internal memory
  and potentially to solve model checking tasks beyond the reach of conventional implementations.

  Compared to the only other external memory BDD package, CAL, Adiar is several orders of magnitude
  faster when computing on larger instances.

\end{abstract}


\section{Introduction} \label{sec:introduction}

Binary Decision Diagrams~\cite{Bryant1986} (BDDs) are a concise and canonical representation of
$n$-ary Boolean functions as directed acyclic graphs. Starting with
\cite{Bose1989,Coudert1990,Burch1992}, BDDs have become a popular tool for symbolic model checking.
To this day, they are still used for model checking
probabilistic~\cite{Kwiatkowska2002,Hensel2022,Fu2022} and multi-agent
systems~\cite{Lomuscio2017,Gammie2004,Su2007,Fu2022} and CTL$^*$ formulas~\cite{Amparore2020}.
Furthermore, they are also still used for verifying network
configurations~\cite{AlShaer2009,AlShaer2010,Lopes2013,Brown2023},
circuits~\cite{Kaivola2009,Kaivola2022:FMCAD,Kaivola2022:HoCA}, and feature
models~\cite{Czarnecki2007,Dubslaff2024}. They also have found recent use for the generation of
extended resolution proofs for SAT and QBF problems \cite{Bryant2021:SAT,Bryant2021:QBF,Bryant2022}.
Furthermore, recent research efforts have made progress on fundamental BDD-based procedures, e.g.\
\cite{Larsen2023,Jakobsen2024}, and the very implementation of BDDs, e.g.\
\cite{Beyer2021,Soelvsten2022:TACAS,Husung2024}

The conciseness of BDDs mitigate the state space explosion problem of model checking. Yet, there is
an inherent lower bound to how much space (or computation time) is needed to meaningfully
encapsulate the state space. Hence, the size of the BDDs grows together with the size and the
complexity of the model under verification. Inevitably, BDDs must outgrow the machine's RAM for some
models. Yet, most implementations use recursion and hash tables for memoisation
\cite{Lind1999,Somenzi2015,Dijk2016,Beyer2021,Husung2024}. Both recursion and the hash tables
introduce cache misses \cite{Klarlund1996,Pastva2023}. As the BDDs outgrow the RAM, these cache
misses turn into memory swaps which slows computation down by several orders of magnitude
\cite{Soelvsten2022:TACAS}. This puts an upper limit on what BDDs can solve in practice.

Unlike conventional BDD implementations, Adiar~\cite{Soelvsten2022:TACAS} is designed to handle BDDs
that are too large for the RAM. To this end, it replaces the conventional approach of memoised
recursion with time-forward processing~\cite{Chiang1995,Arge1995:2} algorithms. Unlike depth-first
recursion \cite{Arge1995:1}, this technique is efficient in the I/O-model~\cite{Aggarwal1987} of
Aggarwal and Vitter. This I/O-efficiency, in turn, allows Adiar in practice to process large BDDs
without being affected by the disk's speed. Using this technique is only at the cost of a small
overhead to its running time \cite{Soelvsten2022:TACAS}.

\subsection{Contributions}

In \cite{Soelvsten2024:arXiv}, we extended the external memory BDD package Adiar with efficient
multi-variable quantification, based on the notion of nested sweeps. The framework in
\cite{Soelvsten2024:arXiv} was evaluated on quantifiers that occur in Quantified Boolean Formulas.
In this paper, we evaluate the effectiveness of nested sweeps in the context of symbolic model
checking. This requires some extensions to Adiar in order to improve the \emph{relational product}.
This operation is needed for the computation of symbolic successors and predecessors of a set of
states. We show in \cref{sec:replace} how monotone variable substitutions can be piggy-backed ``for
free'' onto Adiar's other algorithms. Furthermore, we show in \cref{sec:and-exists} how to combine
the conjunction and quantification operations to further improve the running time of the relational
product. Finally, \cref{sec:conclusion} identifies future work and provides recommendations based on
our experimental results in \cref{sec:experiments}.

\section{Preliminaries} \label{sec:preliminaries}

\subsection{I/O Model}

To make algorithmic analysis tractable, Aggarwal and Vitter's I/O-model~\cite{Aggarwal1987} is an
abstraction of the machine's complex memory hierarchy. This model consists of two levels of memory:
the \emph{internal memory}, e.g.\ the RAM, and the \emph{external memory}, e.g.\ the disk. To
compute on some data, it has to reside in internal memory. Yet, whereas the external memory is
unlimited, the internal memory has only space for $M$ elements. Hence, if the input of size $N$ or
some auxiliary data structure exceeds $M$ then data has to be offloaded to external memory. Yet,
each data transfer between the two levels, i.e. each read and write, consists of $B$ sized
\emph{blocks} of consecutive data. Intuitively, an algorithm is I/O-efficient if it makes sure to
use a substantial portion of the $B$ elements in each block that has been read.

An algorithm's I/O-complexity is the number of data transfers, I/Os, it uses. For example, reading
an input sequentially requires $N/B$ I/Os whereas random access costs up to $N$ I/Os. Furthermore,
one can sort $N > M$ elements in $\Theta(\sort{N}) \triangleq \Theta(N/B \cdot \log_{M/B}(N/B))$
I/Os \cite{Aggarwal1987}. For most realistic values of $N$, $M$, and $B$, $N/B < \sort{N} \ll N$.

\subsection{Binary Decision Diagrams}

Binary Decision Diagrams~\cite{Bryant1986} (BDDs), based on \cite{Lee1959,Akers1978}, represent
Boolean formul{\ae} as singly-rooted binary directed acyclic graphs. These consist of two
sink nodes (\emph{terminals}) with the Boolean values $\top$ and $\bot$. Each non-sink node
(\emph{BDD node}) contains a decision variable, $x_i$, together with two children, $v_{\top}$ and
$v_{\bot}$. Together, these three values represent the if-then-else decision
$x_i\ ?\ v_{\top}\ :\ v_{\bot}$. Hence, the BDD represents an $n$-ary Boolean formula. In
particular, each path from the root to the $\top$ terminal represents one (or more) assignments for
which the function outputs $\top$. For example, \cref{fig:bdd_example:init} represents the formula
$x_1 \land \neg x_2$.

What are colloquially called BDDs are in particular \emph{ordered} and \emph{reduced} BDDs. A BDD is
ordered, if the decision variables occur only once on each path and always according to the same
ordering \cite{Bryant1986}. An ordered BDD is furthermore reduced if it neither contains duplicate
subgraphs nor any BDD nodes have their two children being the same \cite{Bryant1986}. Assuming the
variable ordering is fixed, reduced and ordered BDDs are a canonical representation of a Boolean
formul{\ae} \cite{Bryant1986}.

\begin{figure}[t]
  \centering

  \subfloat[TS]{
    \label{fig:bdd_example:ts}
    \centering

    \begin{tikzpicture}
      \node[] at (-0.8, 4) (init) {};
      \node[shape = circle, draw = black] at (0, 4) (x) {$1$};
      \node[shape = circle, draw = black] at (0, 2) (y) {$2$};

      \path[->] (init) edge (x);
      \path[->] (x) edge[bend left] (y);
      \path[->] (y) edge[bend left] (x);
      \path[->] (y) edge[loop left] (y);
    \end{tikzpicture}

    \vspace{54pt}
  }
  \quad
  \subfloat[BDD]{
    \label{fig:bdd_example:rel}
    \centering

    \begin{tikzpicture}
      \node[shape = circle,    draw = black] at (0, 4) (xu) {$x_1$};

      \node[shape = circle,    draw = black, inner sep=1.8pt] at (-1.0, 3) (xp_1) {$x_1'$};
      \node[shape = circle,    draw = black, inner sep=1.8pt] at ( 1.0, 3) (xp_2) {$x_1'$};

      \node[shape = circle,    draw = black] at (-0.7, 2) (yu_1) {$x_2$};
      \node[shape = circle,    draw = black] at ( 0.7, 2) (yu_2) {$x_2$};

      \node[shape = circle,    draw = black, inner sep=1.8pt] at ( 0.0, 1) (yp) {$x_2'$};

      \node[shape = rectangle, draw = black] at ( 0.0, 0) (T) {$\top$};

      \draw[->, dashed]
        (xu) edge (xp_1)
        (xp_2) edge (yu_2)
        (yu_2) edge (yp)
      ;
      \draw[->]
        (xu) edge (xp_2)
        (xp_1) edge (yu_1)
        (xp_2) edge (yu_1)
        (yu_1) edge (yp)
        (yp) edge (T)
      ;
    \end{tikzpicture}
  }
  \quad
  \subfloat[BDD]{
    \label{fig:bdd_example:init}
    \centering

    \quad
    \begin{tikzpicture}
      \node[shape = circle,    draw = black] at (0.5, 4) (x) {$x_1$};
      \node[shape = circle,    draw = black] at (0.5, 2) (y) {$x_2$};

      \node[shape = rectangle, draw = black] at ( 0.5, 0) (T) {$\top$};

      \draw[->, dashed]
        (y) edge (T)
      ;
      \draw[->]
        (x) edge (y)
      ;
    \end{tikzpicture}
    \quad
  }
  \quad
  \subfloat[BDD]{
    \label{fig:bdd_example:next}
    \centering

    \quad
    \begin{tikzpicture}
      \node[shape = circle,    draw = black] at (0.5, 4) (x) {$x_1$};
      \node[shape = circle,    draw = black] at (0.5, 2) (y) {$x_2$};

      \node[shape = rectangle, draw = black] at ( 0.5, 0) (T) {$\top$};

      \draw[->, dashed]
        (x) edge (y)
      ;
      \draw[->]
        (y) edge (T)
      ;
    \end{tikzpicture}
    \quad
  }
  \quad
  \subfloat[BDD]{
    \label{fig:bdd_example:closure}
    \centering

    \quad
    \begin{tikzpicture}
      \node[shape = circle,    draw = black] at (0.5, 4) (x) {$x_1$};
      \node[shape = circle,    draw = black] at (0.5, 2) (y) {$x_2$};

      \node[shape = rectangle, draw = black] at ( 0.5, 0) (T) {$\top$};

      \draw[->]
        (x) edge (y)
        (y) edge (T)
      ;
    \end{tikzpicture}
    \quad
  }

  \caption{The transition system in \ref{fig:bdd_example:ts} can be represented as the BDD for its
    relation relation \ref{fig:bdd_example:rel} and the BDD for its initial state in
    \ref{fig:bdd_example:init} via a
    \emph{unary} encoding and an \emph{interleaved} variable ordering. One (More) relational product
    of these two BDDs creates the one in \ref{fig:bdd_example:next} (\ref{fig:bdd_example:closure}).\\
    For readability, we suppress the $\bot$ terminal. The $\top$ terminal is drawn as a box while
    BDD nodes are drawn as circles surrounding their decision variable. The \emph{then}, resp.\
    \emph{else}, child is drawn solid, resp.\ dashed.}
  \label{fig:bdd_example}
\end{figure}
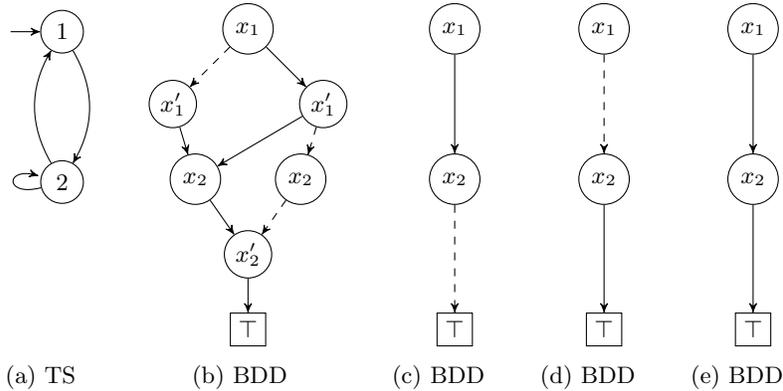

\subsubsection{Relational Product}

In the case of symbolic model checking, one or more BDDs, $R_{\vec{x}, \vec{x}'}$, represent
relations between unprimed Boolean variables, $\vec{x}$, that encode the \emph{current} and primed
variables, $\vec{x}'$, that encode the \emph{next} state. These relations are applied to sets of
states, $S_{\vec{x}}$, which is identified using only the unprimed variables. For example, the
transition system in \cref{fig:bdd_example:ts}, can be represented via a unary encoding with the two
variables $x_1$ and $x_2$ which represent the states $1$ and $2$, respectively. The transitions
would then be the relation in \cref{fig:bdd_example:rel}. In particular, \cref{fig:bdd_example:rel}
represents each of the three transitions in \cref{fig:bdd_example:ts} as the disjunction of the
following three formulas.
\begin{align*}
  x_1 \land \neg x_1' \land \neg x_2 \land x_2'
  &&
  \neg x_1 \land x_1' \land x_2 \land \neg x_2'
  &&
  \neg x_1 \land \neg x_1' \land x_2 \land x_2'
\end{align*}
As \cref{fig:bdd_example:rel} is a disjunction of all three transitions, it is a
joint~\cite{Burch1994} relation\footnote{This entire example only pertains to a unary encoding of a
  transition system with asynchronous semantics. If the transition system is synchronous, then the
  conjunction would be used instead and one has to join all transitions together into a single
  relation.}. Instead, one could also keep \cref{fig:bdd_example:rel} as three separate BDDs for a
disjoint~\cite{Burch1994} relation. This has the benefit of representing the transition system
symbolically via smaller BDDs at the cost of having to apply the transitions one-by-one
\cite{Burch1994}. The initial state of the transition system would be the BDD in
\cref{fig:bdd_example:init}, i.e.\ the formula $x_1 \land \neg x_2$.

These BDDs are manipulated using the following two operations (the \emph{relational product}) to
obtain the next or the previous set of states.
\begin{align}
  \label{eq:relnext}
  \mathit{Next}(S_{\vec{x}}, R_{\vec{x}, \vec{x}'})
  &\triangleq (\exists \vec{x} : S_{\vec{x}} \land R_{\vec{x}, \vec{x}'})[\vec{x}' / \vec{x}]
  \\
  \label{eq:relprev}
  \mathit{Prev}(S_{\vec{x}}, R_{\vec{x}, \vec{x}'})
  &\triangleq \exists \vec{x}' : S_{\vec{x}}[\vec{x} / \vec{x}'] \land R_{\vec{x}, \vec{x}'}
\end{align}
For example, $\mathit{Next}$ of \cref{fig:bdd_example:init} and \cref{fig:bdd_example:rel} is the
BDD shown in \cref{fig:bdd_example:next}. The transitive closure of applying $\mathit{Next}$, i.e.\
the result of $\mathit{Next}^*$ of \cref{fig:bdd_example:init} and \cref{fig:bdd_example:rel}, is
shown in \cref{fig:bdd_example:closure}.

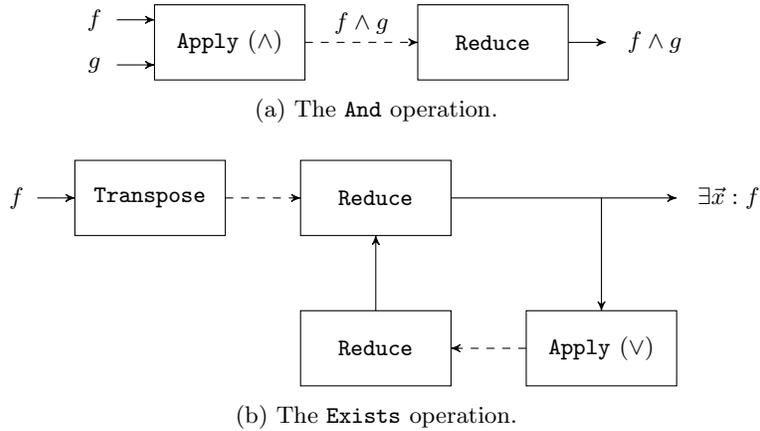
\begin{figure}[t]
  \centering

  \subfloat[The \texttt{And} operation.] {
    \label{fig:tandem:and}

    \begin{tikzpicture}[every text node part/.style={align=center}]

      \draw (0,0) rectangle ++(2,1)
      node[pos=.5]{\texttt{Apply} ($\land$)};
      \draw (3.5,0) rectangle ++(2,1)
      node[pos=.5]{\texttt{Reduce}};

      \draw[->] (-0.5,0.8) -- ++(0.5,0)
      node[pos=-0.6]{$f$};
      \draw[->] (-0.5,0.2) -- ++(0.5,0)
      node[pos=-0.6]{$g$};

      \draw[->,dashed] (2,0.5) -- ++(1.5,0)
      node[pos=0.5,above]{\small $f \land g$};

      \draw[->] (5.5,0.5) -- ++(0.5,0)
      node[pos=2.3]{$f \land g$};
    \end{tikzpicture}
  }

  \bigskip

  \subfloat[The \texttt{Exists} operation.] {
    \label{fig:tandem:exists}

    \begin{tikzpicture}[every text node part/.style={align=center}]
      \draw (0,0) rectangle ++(2,1)
      node[pos=.5]{\texttt{Transpose}};
      \draw (3,0) rectangle ++(2,1)
      node[pos=.5]{\texttt{Reduce}};

      \draw[->] (-0.5,0.5) -- ++(0.5,0)
      node[pos=-0.55]{$f$};

      \draw[->, dashed] (2,0.5) -- ++(1,0);

      \draw[->] (5,0.5) -- ++(3,0)
      node[pos=1.23]{$\exists \vec{x} : f$};

      \draw (6,-2) rectangle ++(2,1)
      node[pos=.5]{\texttt{Apply} ($\lor$)};
      \draw (3,-2) rectangle ++(2,1)
      node[pos=.5]{\texttt{Reduce}};

      \draw[->] (7,0.5) -- ++(0,-1.5);

      \draw[->, dashed] (6,-1.5) -- ++(-1,0);
      \draw[->] (4,-1) -- ++(0,1);
    \end{tikzpicture}
  }

  \caption{The \texttt{Apply}--\texttt{Reduce} pipelines in Adiar.}
  \label{fig:tandem}
\end{figure}

\subsection{I/O-efficient BDD Manipulation}

The Adiar~\cite{Soelvsten2022:TACAS} BDD package (based on \cite{Arge1995:1}) aims to compute
efficiently on BDDs that are so large they have to be stored on the disk. To do so, rather than
using depth-first recursion, it processes the BDDs level by level with time-forward
processing~\cite{Chiang1995,Arge1995:2}. This makes it optimal in the I/O-model \cite{Arge1995:1}.
As shown in \cref{fig:tandem:and}, the basic BDD operations, such as the \texttt{And} ($\land$), are
computed in two phases. First, the resulting, but not necessarily reduced, BDD is computed in a
top-down \texttt{Apply} sweep \cite{Soelvsten2022:TACAS}. Secondly, this BDD is made canonical in a
bottom-up \texttt{Reduce} sweep \cite{Soelvsten2022:TACAS}. As shown in \cref{fig:tandem:exists},
more complex BDD operations, such as \texttt{Exists} ($\exists$), are computed by accumulating the
result of multiple nested \texttt{Apply}--\texttt{Reduce} sweeps bottom-up in an outer
\texttt{Reduce} sweep \cite{Soelvsten2024:arXiv}. Here, each set of nested
\texttt{Apply}--\texttt{Reduce} sweeps computes the \texttt{Or} ($\lor$) of a to-be quantified
variable's cofactors \cite{Soelvsten2024:arXiv}.

\section{I/O-efficient Relational Product} \label{sec:relprod}

\subsection{I/O-efficient Variable Substitution} \label{sec:replace}

The work in \cite{Soelvsten2022:TACAS,Soelvsten2024:arXiv} covers all BDD operations in
\cref{eq:relnext,eq:relprev} but the variable substitution ($[\vec{x}' / \vec{x}]$) between primed
and unprimed variables. Hence, this operation's design is the last step towards an I/O-efficient
relational product.

In general, the substitution algorithm is equivalent to changing the variable ordering. Yet, this
operation is notorious for being hard to compute since it greatly affects the structure and the size
of the BDD. But, in the context of symbolic model checking, it merely suffices to consider variable
substitutions, $\pi$, that are \emph{monotone} with respect to the variable ordering. That is, if a
variable $x_i$ is prior to $x_j$ in the given BDD then $\pi(x_i)$ is also prior to $\pi(x_j)$ in the
substituted BDD. For example, the BDDs in \cref{fig:bdd_example} use a variable ordering where
primed and unprimed variables are interleaved and so the variable substitutions in
\cref{eq:relnext,eq:relprev} are monotone. This monotonicity is useful, since the BDDs before and
after such a substitution are isomorphic.

\begin{proposition}
  \label{thm:naive}

  A monotone variable substitution can be applied to a BDD of $N$ nodes in $\Oh{N}$ time and using
  $2 \cdot \tfrac{N}{B}$ I/Os.
\end{proposition}
\begin{proof}
  Using $\tfrac{N}{B}$ I/Os, one can stream through the $N$ BDD nodes in external memory and apply
  $\pi$ onto all BDD nodes requiring $\Oh{N}$ time. The result is simultaneously streamed back to
  external memory using another $\tfrac{N}{B}$ I/Os. \qed
\end{proof}

This is optimal if the input is reduced and the output needs to exist separately in external memory.
Yet, if the input still needs to be reduced, then substitution can be integrated into the
\texttt{Reduce} algorithm from \cite{Soelvsten2022:TACAS} as follows: when reducing the level for
variable $x_i$, output its new nodes with variable $\pi(x_i)$.
\begin{proposition}
  \label{thm:reduce}
  A monotone variable substitution can be applied during the \texttt{Reduce} sweep to an unreduced
  BDD with $n$ levels in $\Oh{n}$ time, using no additional space in internal memory, and not using
  any additional I/Os.
\end{proposition}
Particular to \cref{eq:relnext}, variable substitution can be integrated into the quantification
operation's outer \texttt{Reduce} sweep which preceedes it.

Similarly, substitution in \cref{eq:relprev} can become part of the conjunction operation that
succeeds it. Here, each BDD node of $S_{\vec{x}}$ would be mapped on-the-fly as they are streamed
from external memory. The implementation of such an idea can be greatly simplified with one
additional restriction on $\pi$. Note, in the context of model checking, variable substitutions are
not only monotone but also \emph{affine}, i.e.\ $\pi(x_i) = x_{\alpha \cdot i + \beta}$ for some
$\alpha, \beta \in \N$ and $\alpha \geq 1$. Hence, by storing $\alpha$ and $\beta$ in internal
memory, one can defer applying $\pi$ until they are read as part of the succeeding BDD operation.
This makes variable substitution a constant-time operation by adding an overhead to all other BDD
operations.
\begin{proposition}
  \label{thm:affine} \label{thm:shift}

  A monotone and affine variable substitution can be applied to a BDD in $\Oh{1}$ time, using
  $\Oh{1}$ additional space in internal memory, and using no additional I/Os.
\end{proposition}
In practice, $\alpha$ is always $1$. Hence, only $\beta$ needs to be stored.

\subsection{An I/O-efficient \texttt{AndExists}} \label{sec:and-exists}

Just as \cref{thm:reduce,thm:affine} move the variable substitution inside another operation, the
relational product's performance can be further improved by merging the conjunction and existential
quantification into a single \texttt{AndExists} \cite{Dijk2015}.

The quantification operation's outer \texttt{Reduce} sweep requires the input to be transposed
\cite{Soelvsten2024:arXiv}. Hence, \cref{fig:tandem:exists} includes an initial transposition step.
The \texttt{Apply} step of the \texttt{And} in \cref{fig:tandem:and} produces an unreduced output
which already is transposed \cite{Soelvsten2022:TACAS}. This means that naively combining
\cref{fig:tandem:and,fig:tandem:exists} into an \texttt{AndExists} will result in some computational
steps having no effect: the reduced BDD from the \texttt{And} is immediately transposed and reduced
once more as part of the \texttt{Exists}. Hence, the \texttt{Reduce} step of the \texttt{And} and
the \texttt{Transpose} step of the \texttt{Exists} can be skipped.

\begin{optimisation}
  \label{opt:and-transpose}

  Use the unreduced result of the conjunction operation as the input to the quantification
  operation's outer \texttt{Reduce} sweep.
\end{optimisation}
This makes the quantification operation's outer \texttt{Reduce} sweep also do double duty as the
conjunction operation's \texttt{Reduce} sweep. This saves the linearithmic time and I/Os otherwise
needed to first reduce the result of the conjunction and to then transpose it.

Furthermore, experiments in \cite{Soelvsten2024:arXiv} indicate, the quantification algorithm can
become up to $\sim 21\%$ faster by pruning a BDD node (and possibly its subtrees) where
quantification is trivial because one of or both its children are terminals.

\begin{optimisation}
  \label{opt:prune}

  During the conjunction's \texttt{Apply} sweep, prune the resulting BDD nodes that have to-be
  quantified variables.
\end{optimisation}

\section{Experimental Evaluation} \label{sec:experiments}

We have added the \texttt{bdd\_replace} function to Adiar to provide support for variable
substitution. For now, it only supports monotone substitutions with the three propositions presented
in \cref{sec:replace}. Building on this, we added the \texttt{bdd\_relprod}, \texttt{bdd\_relnext},
and \texttt{bdd\_relprev} operations for model checking applications. This includes the ideas in
\cref{sec:and-exists}. Optimisation~\ref{opt:and-transpose} was added in its general form to both
the \texttt{bdd\_exists} and the \texttt{bdd\_forall} functions: transposition is skipped if the
input BDD is unreduced.

With this in hand, we have conducted multiple experiments to evaluate the impact and utility of this
work. In particular, we have sought to answer the following research questions:

\begin{enumerate}
\item\label{rq:optimisations}

  What is the impact of the proposed optimisations in \cref{sec:relprod}?

\item\label{rq:competitors:depth-first}

  How does Adiar's Relational Product operation compare to conventional depth-first implementations?

\item\label{rq:competitors:breadth-first}

  How does Adiar's Relational Product operation compare to the breadth-first algorithms of CAL?
\end{enumerate}

\subsection{Benchmarks}

We have extended our BDD benchmarking
suite\footnote{\href{https://github.com/SSoelvsten/bdd-benchmark}{github.com/SSoelvsten/bdd-benchmark}}
\cite{Soelvsten2022:TACAS} with the foundations for a symbolic model checker for Petri
Nets~\cite{Petri1967} and Asynchronous Boolean Networks~\cite{Kauffman1969,Thomas1991}. This
benchmark explores the given model symbolically as follows:

\begin{itemize}
\item\emph{Reachability}:

  The set of reachable states, $S_{\text{reach}}$, is computed via the transitive closure
  $\mathit{Next}^*(S_I, R)$ on the initial state, $S_I$, and transition relation $R$.
  This uses \texttt{bdd\_relnext} up to a polynomial number of times with respect to the model and
  its state space.

  \smallskip

\item\emph{Deadlock}:

  The set of deadlocked states are identified via
  $S_{\text{reach}} \setminus \mathit{Prev}(S_{\text{reach}}, R)$.
  This requires a single use of \texttt{bdd\_relprev} if a joint partitioning~\cite{Burch1994} is
  used. If a disjoint partitioning~\cite{Burch1994} is used then this operation is called once for
  each transition in the model.
\end{itemize}

\noindent Based on \cite{Meijer2015}, the variable ordering is predetermined by analysing the given
model\footnote{Even though \cite{Meijer2015} suggests one runs the algorithm on a bipartite
  read/write graph derived from the model, we instead run it on an incidence graph derived from the
  model. Our preliminary experiments indicate this further decreases the size of the BDDs.} with
Sloan's algorithm~\cite{Sloan1989}. For each model, we have run them with both a joint and a
disjoint~\cite{Burch1994} partitioning of the transition relation.

Based on preliminary experiments, we identified 75 model instances that were solvable with Adiar. In
particular, these are 16 Petri nets from the 2021--2023 Model Checking
Competitions~\cite{Kordon2021,Kordon2022,Kordon2023}, 41 Boolean networks distributed with
AEON~\cite{Benes2020}, and 18 Boolean networks distributed with PyBoolNet~\cite{Klarner2016}.

For all 150 of these instances, none of the BDDs generated during either benchmark grew larger than
$2 \cdot 10^6$ BDD nodes ($48$~MiB). Furthermore, the BDDs encoding the respective initial states
required only 294 or fewer BDD nodes ($6.9$~KiB). While computing the transition relation's
transitive closure, the BDDs grew slowly. That is, most, if not all, BDD computations in either
benchmark are too small to be within the current scope of Adiar \cite{Soelvsten2023:ATVA}. Hence,
inspired by the recent work of Pastva and Henzinger \cite{Pastva2023}, we have also created the
following two additional benchmarks.


\begin{itemize}
\item\emph{Next}:

  Given a BDD with a set of states, $S_{\vec{x}}$, and another one with a relation,
  $R_{\vec{x}, \vec{x}'}$, the next set of states,
  $\mathit{Next}(S_{\vec{x}}, R_{\vec{x}, \vec{x}})$, is computed with a single
  \texttt{bdd\_relnext}.

  \smallskip

\item\emph{Prev}:

  Given a BDD with a set of states, $S_{\vec{x}}$, and another one with a relation,
  $R_{\vec{x}, \vec{x}'}$, the previous set of states,
  $\mathit{Prev}(S_{\vec{x}}, R_{\vec{x}, \vec{x}})$, is computed with a
  \texttt{bdd\_relprev}.
\end{itemize}

For inputs, we have followed the approach in \cite{Pastva2023}. The above-described reachability
analysis has been extended to save the (joint) transition relation BDD, $R_{\vec{x}, \vec{x}'}$,
together with the first state BDD, $S_{\vec{x}}$, constructed of each order of magnitude. Using
this, we have generated BDDs by running reachability analysis on all models from the 2020--2023
Model Checking Competitions~\cite{Kordon2020,Kordon2021,Kordon2022,Kordon2023}. This was done using
LibBDD~\cite{Benes2020} as the BDD backend and a time limit of 1~h on a Ubuntu 24.4 machine with a
12-core $3.6$~GhZ Intel i7-12700 processor and 64~GiB of memory. This has resulted in serialised
BDDs from 124 model instances. These BDDs have been made publically available at the following DOI:

\begin{center}
  \doi{10.5281/zenodo.13928216}
\end{center}

For our evaluation, we focus on the three models \texttt{GPUForwardProgress} 20a (2021),
\texttt{SmartHome} 16 (2020), and \texttt{ShieldPPPs} 10a (2020) where we could generate a set of
states with a magnitude of $2^{25}$ BDD nodes\footnote{A fourth model, \texttt{SmartHome} 17 (2020),
  also created a set of states this large. Yet, the serialized BDD for the relation turned out to be
  corrupted. The published set of BDDs above has fixed this and other data corruptions. Furthermore,
  the repository also includes large BDDs that required more than 1~h to be generated.}. The
relation sizes are 500~MiB, 270~MiB, and 0.1~MiB, respectively. For state size, we focus on the four
largest orders of magnitude constructed, i.e.\ from $2^{22}$ (40~MiB) to $2^{25}$ BDD nodes
(320~MiB).

\subsection{Hardware, Settings, and Measurements}

Similar to
\cite{Soelvsten2022:TACAS,Soelvsten2023:NFM,Soelvsten2023:ATVA,Soelvsten2024:SPIN,Soelvsten2024:arXiv},
we ran our experiments on machines at the Centre for Scientific Computing in Aarhus. These machines
run Rocky Linux (Kernel 4.18.0-513) with 48-core $3.0$ GHz Intel Xeon Gold 6248R processors,
$384$~GiB of RAM, $3.5$~TiB of SSD disk (with $48$~GiB of swap memory). The benchmark and BDD
packages were compiled with GCC $10.1.0$ and Rust $1.72.1$. Due to an update to the cluster's
machines, LibBDD~\cite{Benes2020} was compiled with GCC $13.2.0$ and Rust $1.77.1$ for the
\emph{Next} and \emph{Prev} benchmarks.

Each BDD package was given $9/10$th of these $384$~GiB of internal memory\footnote{CUDD and LibBDD
  ignore any given memory limit and hence may have used more.}; this leaves the remaining space for
the OS and the benchmark's other data structures. Otherwise, each BDD package was given a single CPU
core and initialised with their respective default and/or recommended settings.

For \emph{Reachability} and \emph{Deadlock} (\cref{tab:versus,fig:versus.mcnet}), we have measured
the running time 3 times with a timeout of 48~h. For \emph{Next}, and \emph{Prev}
(\cref{tab:versus,fig:versus.relprod}) we measured the running time of each instance 5 times and
bounded the running time to 12~h. For RQ~\ref{rq:optimisations}, we tried to minimise the noise due
to hardware and the OS in the reported numbers. Similar to
\cite{Soelvsten2022:TACAS,Soelvsten2023:ATVA,Soelvsten2024:arXiv} (based on \cite{Chen2016}), we
intended to do so by reporting the smallest measured running time. But, some measurements seem to
have been taken while the machines were in a particularly good state. Using the minimum in this case
makes RQ~\ref{rq:optimisations} much harder to investigate. Hence, we instead resort to reporting
the median.

For the data shown in \cref{fig:memory} for RQ~\ref{rq:competitors:depth-first}, we deemed that a
single measurement of each data point would suffice.

\subsection{RQ~\ref{rq:optimisations}: Effect of the Optimisations}

We have implemented the ideas from \cref{sec:replace} and \cref{sec:and-exists} in order of their
complexity and expected benefit in practice: \cref{thm:naive}~(---),
Optimisation~\ref{opt:and-transpose}~(\markSkipTranspose),
Optimisation~\ref{opt:prune}~(\markPruningAnd), \cref{thm:reduce}~(\markExistsReplace), and finally
\cref{thm:affine}~(\markShiftReplace). For evaluation, we have measured the running time of these
five accumulated set of features. \Cref{fig:versions.mcnet,fig:versions.relprod} show the relative
performance increase compared to only using \cref{thm:naive}.

\begin{table}[b]
  \centering

  \caption{Total running time (seconds) of each version of Adiar. The \# column indicates the number
    of instances that were solved by all five versions.}
  \label{tab:versions}

  \bgroup
  \def\arraystretch{1.1}
  \setlength\tabcolsep{5pt}

  \begin{tabular}{l||c||rrrrr}
                          & \#  & {---$^{\text{Prop.~\ref{thm:naive}}}$}
                                             & {\markSkipTranspose$^{\text{Prop.~\ref{thm:naive}}}_{\text{Opt.~\ref{opt:and-transpose}}}$}
                                                          & {\markPruningAnd$^{\text{Prop.~\ref{thm:naive}}}_{\text{Opt.~\ref{opt:and-transpose}+\ref{opt:prune}}}$}
                                                                       & {\markExistsReplace$^{\text{Prop.~\ref{thm:naive}+\ref{thm:reduce}}}_{\text{Opt.~\ref{opt:and-transpose}+\ref{opt:prune}}}$}
                                                                                    & {\markShiftReplace$^{\text{Prop.~\ref{thm:naive}+\ref{thm:reduce}+\ref{thm:affine}}}_{\text{Opt.~\ref{opt:and-transpose}+\ref{opt:prune}}}$}
    \\ \hline \hline
    \emph{Reachability}   & 147 & 4134.2     & 3918.1     & 3738.6     & 3627.5     & 3659.7
    \\
    \emph{Deadlock}       & 147 & 416.4      & 269.3      & 246.5      & 248.1      & 223.8
    \\ \hline
    \emph{Next}           & 12  & 24921.8    & 24378.0    & 23994.6    & 23257.7    & 23628.1
    \\
    \emph{Prev}           & 10  & 77541.9    & 78068.2    & 76862.0    & 75693.8    & 76353.8
  \end{tabular}
  \egroup
\end{table}

\begin{figure}[b!]
  \centering

  \subfloat[\emph{Reachability}] {
    \label{fig:versions.mcnet:reachability}

    \begin{tikzpicture}
      \begin{axis}[%
        width=0.48\linewidth, height=0.4\linewidth,
        every tick label/.append style={font=\scriptsize},
        xlabel={\scriptsize Running Time with Prop.~\ref{thm:naive} (s)},
        xmin=0.05,
        xtick={0.01,0.1,1,10,100,1000,10000},
        xmax=3000,
        xmode=log,
        ylabel={\scriptsize Speed-up},
        ymin=0.8,
        ymax=1.8,
        ytick={0.8, 1, 1.2, 1.4, 1.6, 1.8},
        yticklabels={
          0.8 $\times$,
          1.0 $\times$,
          1.2 $\times$,
          1.4 $\times$,
          1.6 $\times$,
          1.8 $\times$,
        },
        grid style={dashed,black!12},
        ]

        \addplot[domain=0.01:3000, samples=8, color=black]
        {1};

        \begin{scope}[blend mode=soft light]
          \addplot+ [style=dots_adiar-skip_transpose]
          table {./data/reachability.adiar-skip_transpose.tex};

          \addplot+ [style=dots_adiar-pruning_and]
          table {./data/reachability.adiar-pruning_and.tex};

          \addplot+ [style=dots_adiar-exists_replace]
          table {./data/reachability.adiar-exists_replace.tex};

          \addplot+ [style=dots_adiar-shift_replace]
          table {./data/reachability.adiar-shift_replace.tex};

          \addplot[domain=0.05:10000, samples=8, style=dashed, opacity=0.5, color=adiar-skip_transpose]
          {1.1320520920371682};
          \addplot[domain=0.05:10000, samples=8, style=dashed, opacity=0.5, color=adiar-pruning_and]
          {1.1561359418512216};
          \addplot[domain=0.05:10000, samples=8, style=dashed, opacity=0.5, color=adiar-exists_replace]
          {1.1929371164022806};
          \addplot[domain=0.05:10000, samples=8, style=dashed, opacity=0.5, color=red]
          {1.1849073763432338};
        \end{scope}
      \end{axis}
    \end{tikzpicture}
  }
  \subfloat[\emph{Deadlocks}] {
    \label{fig:versions.mcnet:deadlock}

    \begin{tikzpicture}
      \begin{axis}[%
        width=0.48\linewidth, height=0.4\linewidth,
        every tick label/.append style={font=\scriptsize},
        xlabel={\scriptsize Running Time with Prop.~\ref{thm:naive} (s)},
        xmin=0.005,
        xtick={0.001,0.01,0.1,1,10,100,1000,10000},
        xmax=3000,
        xmode=log,
        ylabel={\scriptsize Speed-up},
        ymin=0.5,
        ymax=4.5,
        ytick={1, 2.0, 3.0, 4.0},
        yticklabels={
          1 $\times$,
          2 $\times$,
          3 $\times$,
          4 $\times$,
        },
        grid style={dashed,black!12},
        ]

        \addplot[domain=0.001:3000, samples=8, color=black]
        {1};

        \begin{scope}[blend mode=soft light]
          \addplot+ [style=dots_adiar-skip_transpose]
          table {./data/deadlock.adiar-skip_transpose.tex};

          \addplot+ [style=dots_adiar-pruning_and]
          table {./data/deadlock.adiar-pruning_and.tex};

          \addplot+ [style=dots_adiar-exists_replace]
          table {./data/deadlock.adiar-exists_replace.tex};

          \addplot+ [style=dots_adiar-shift_replace]
          table {./data/deadlock.adiar-shift_replace.tex};

          \addplot[domain=0.005:10000, samples=8, style=dashed, opacity=0.5, color=adiar-skip_transpose]
          {1.084506429063764};
          \addplot[domain=0.005:10000, samples=8, style=dashed, opacity=0.5, color=adiar-pruning_and]
          {1.1304087654038872};
          \addplot[domain=0.005:10000, samples=8, style=dashed, opacity=0.5, color=adiar-exists_replace]
          {1.093934702445302};
          \addplot[domain=0.005:10000, samples=8, style=dashed, opacity=0.5, color=red]
          {1.099224638982916};
        \end{scope}
      \end{axis}
    \end{tikzpicture}
  }

  \smallskip

  \begin{tikzpicture}
    \begin{customlegend}[
      legend columns=-1,
      legend style={draw=none,column sep=1ex},
      legend entries={
        {$^{\text{Prop.~\ref{thm:naive}}}$},
        {$^{\text{Prop.~\ref{thm:naive}}}_{\text{Opt.~\ref{opt:and-transpose}}}$},
        {$^{\text{Prop.~\ref{thm:naive}}}_{\text{Opt.~\ref{opt:and-transpose}+\ref{opt:prune}}}$},
        {$^{\text{Prop.~\ref{thm:naive}+\ref{thm:reduce}}}_{\text{Opt.~\ref{opt:and-transpose}+\ref{opt:prune}}}$},
        {$^{\text{Prop.~\ref{thm:naive}+\ref{thm:reduce}+\ref{thm:affine}}}_{\text{Opt.~\ref{opt:and-transpose}+\ref{opt:prune}}}$},
      }
      ]
      \addlegendimage{style=dots_adiar-naive}
      \addlegendimage{style=dots_adiar-skip_transpose}
      \addlegendimage{style=dots_adiar-pruning_and}
      \addlegendimage{style=dots_adiar-exists_replace}
      \addlegendimage{style=dots_adiar-shift_replace}
    \end{customlegend}
  \end{tikzpicture}

  \caption{Impact of optimisations on model checking tasks running time. Averages are drawn as
    dashed lines.}
  \label{fig:versions.mcnet}
\end{figure}
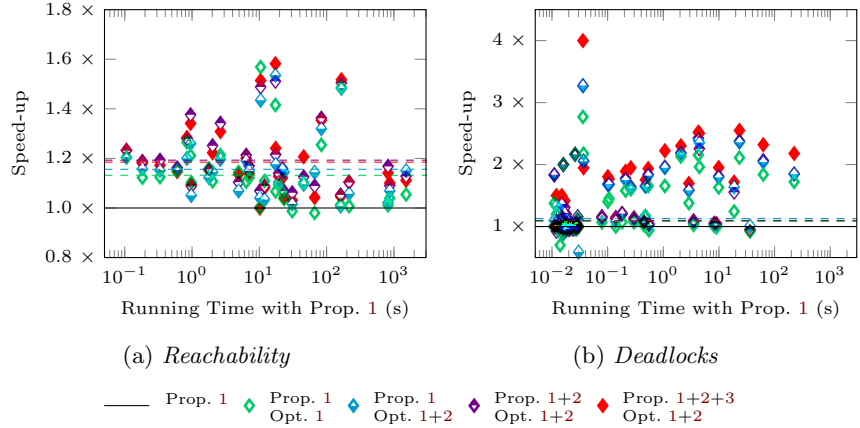

\begin{figure}
  \centering

  \subfloat[\emph{Next}] {
    \label{fig:versions.relprod:next}

    \begin{tikzpicture}
      \begin{axis}[%
        width=0.48\linewidth, height=0.4\linewidth,
        every tick label/.append style={font=\scriptsize},
        xlabel={\scriptsize Running Time with Prop.~\ref{thm:naive} (s)},
        xmin=100,
        xtick={10,100,1000,10000},
        xmax=20000,
        xmode=log,
        ylabel={\scriptsize Speed-up},
        ymin=0.9,
        ymax=1.3,
        ytick={0.9, 1, 1.1, 1.2, 1.3},
        yticklabels={
          0.9 $\times$,
          1.0 $\times$,
          1.1 $\times$,
          1.2 $\times$,
          1.3 $\times$,
        },
        grid style={dashed,black!12},
        ]

        \addplot[domain=0.01:1000000, samples=8, color=black]
        {1};

        \begin{scope}[blend mode=soft light]
          \addplot+ [style=dots_adiar-skip_transpose]
          table {./data/next.adiar-skip_transpose.tex};

          \addplot+ [style=dots_adiar-pruning_and]
          table {./data/next.adiar-pruning_and.tex};

          \addplot+ [style=dots_adiar-exists_replace]
          table {./data/next.adiar-exists_replace.tex};

          \addplot+ [style=dots_adiar-shift_replace]
          table {./data/next.adiar-shift_replace.tex};

          \addplot[domain=100:20000, samples=8, style=dashed, opacity=0.5, color=adiar-skip_transpose]
          {1.037722530705529};
          \addplot[domain=100:20000, samples=8, style=dashed, opacity=0.5, color=adiar-pruning_and]
          {1.058297481489488};
          \addplot[domain=100:20000, samples=8, style=dashed, opacity=0.5, color=adiar-exists_replace]
          {1.0865219867524853};
          \addplot[domain=100:20000, samples=8, style=dashed, opacity=0.5, color=red]
          {1.069799662890797};
        \end{scope}
      \end{axis}
    \end{tikzpicture}
  }
  \subfloat[\emph{Prev}] {
    \label{fig:versions.relprod:prev}

    \begin{tikzpicture}
      \begin{axis}[%
        width=0.48\linewidth, height=0.4\linewidth,
        every tick label/.append style={font=\scriptsize},
        xlabel={\scriptsize Running Time with Prop.~\ref{thm:naive} (s)},
        xmin=100,
        xtick={10,100,1000,10000,100000},
        xmax=100000,
        xmode=log,
        ymin=0.9,
        ymax=1.3,
        ytick={0.9, 1, 1.1, 1.2, 1.3},
        yticklabels={,,,,},
        grid style={dashed,black!12},
        ]

        \addplot[domain=0.01:1000000, samples=8, color=black]
        {1};

        \begin{scope}[blend mode=soft light]
          \addplot+ [style=dots_adiar-skip_transpose]
          table {./data/prev.adiar-skip_transpose.tex};

          \addplot+ [style=dots_adiar-pruning_and]
          table {./data/prev.adiar-pruning_and.tex};

          \addplot+ [style=dots_adiar-exists_replace]
          table {./data/prev.adiar-exists_replace.tex};

          \addplot+ [style=dots_adiar-shift_replace]
          table {./data/prev.adiar-shift_replace.tex};

          \addplot[domain=100:1000000, samples=8, style=dashed, opacity=0.5, color=adiar-skip_transpose]
          {1.010687771310793};
          \addplot[domain=100:1000000, samples=8, style=dashed, opacity=0.5, color=adiar-pruning_and]
          {1.031185743248829};
          \addplot[domain=100:1000000, samples=8, style=dashed, opacity=0.5, color=adiar-exists_replace]
          {1.0452715683052805};
          \addplot[domain=100:1000000, samples=8, style=dashed, opacity=0.5, color=adiar-shift_replace]
          {1.0374409648708278};
        \end{scope}
      \end{axis}
    \end{tikzpicture}
  }

  \smallskip

  \begin{tikzpicture}
    \begin{customlegend}[
      legend columns=-1,
      legend style={draw=none,column sep=1ex},
      legend entries={
        {$^{\text{Prop.~\ref{thm:naive}}}$},
        {$^{\text{Prop.~\ref{thm:naive}}}_{\text{Opt.~\ref{opt:and-transpose}}}$},
        {$^{\text{Prop.~\ref{thm:naive}}}_{\text{Opt.~\ref{opt:and-transpose}+\ref{opt:prune}}}$},
        {$^{\text{Prop.~\ref{thm:naive}+\ref{thm:reduce}}}_{\text{Opt.~\ref{opt:and-transpose}+\ref{opt:prune}}}$},
        {$^{\text{Prop.~\ref{thm:naive}+\ref{thm:reduce}+\ref{thm:affine}}}_{\text{Opt.~\ref{opt:and-transpose}+\ref{opt:prune}}}$},
      }
      ]
      \addlegendimage{style=dots_adiar-naive}
      \addlegendimage{style=dots_adiar-skip_transpose}
      \addlegendimage{style=dots_adiar-pruning_and}
      \addlegendimage{style=dots_adiar-exists_replace}
      \addlegendimage{style=dots_adiar-shift_replace}
    \end{customlegend}
  \end{tikzpicture}

  \caption{Impact of optimisations on running time of a single relational product.
    Averages are drawn as dashed lines.}
  \label{fig:versions.relprod}
\end{figure}
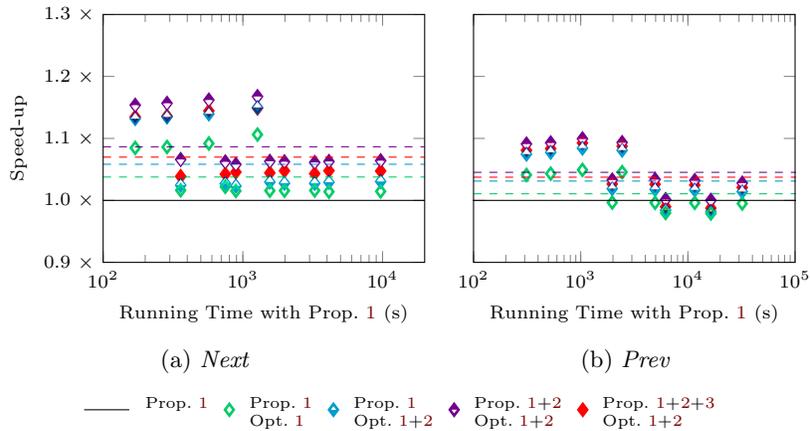

\Cref{tab:versions} shows for each benchmark the total running time of each optimisation. The
running time of the two model checking tasks, \emph{Reachability} and \emph{Deadlock}, are mainly
affected by the optimisations. On the other hand, the running time of \emph{Next} and \emph{Prev}
are comparatively unaffected. This suggest that, as the size of BDDs increases, the computational
cost of variable substitution and the conjunction, which the improvements pertain to, decreases in
comparison to the cost of quantifying variables.

Optimisation~\ref{opt:and-transpose}~(\markSkipTranspose) improves the total running time between
$-0.7\%$ and $35.3\%$ depending on the benchmark. In particular, it mainly improves the running for
the smaller BDDs in \emph{Reachability} and \emph{Deadlock}.

Both the relation, $R_{\vec{x}, \vec{x}'}$, and states, $S_{\vec{x}}$, have many edges to $\bot$
that shortcut the $\land$ operator and gives Optimisation~\ref{opt:prune}~(\markPruningAnd) ample
opportunity to prune subtrees. Hence, it improves the total running time between $4.6\%$ and
$15.7\%$ depending on the benchmark. In particular, for the \emph{Next} benchmark with
\texttt{ShieldPPPs} of a magnitude $2^{24}$, this optimisation decreases the intermediate BDD size
by $30\%$. Yet, the number of nodes processed during the quantification operation's nested inner
\texttt{Apply}--\texttt{Reduce} sweeps is unaffected. Hence, this optimisation may only improve
performance of the conjunction's \texttt{Apply} and \texttt{Reduce} sweeps (the latter of which is
merged with the quantification operation's outer \texttt{Reduce} sweep due to
Optimisation~\ref{opt:and-transpose}~(\markSkipTranspose)) rather than save work during existential
quantification.

\Cref{thm:reduce}~(\markExistsReplace) further improves the total running time for
\emph{Reachability} and \emph{Next} by $3.0\%$ and $3.1\%$, respectively.

Finally, \cref{thm:affine}~(\markShiftReplace) has no effect on the \emph{Reachability} and the
\emph{Next} benchmarks. \cref{tab:versions} shows it adds an overhead of up to $1.6\%$. This
slowdown is due to mapping all BDD nodes on-the-fly with an $\alpha = 1$ and $\beta = 0$, i.e.\
without any actual changes. This overhead can be circumvented by only incorporating this remapping
into the $\land$-operation in \emph{Prev}. Yet, doing so would require a substantially higher
implementation effort and it would also remove the ability to use this optimisation elsewhere. On
the other hand, \cref{tab:versions} also shows \cref{thm:affine}~(\markShiftReplace) improves
\emph{Deadlock} by $9.8\%$ when compared to the version only including up to
\cref{thm:reduce}~(\markExistsReplace). For \emph{Prev}, there is a slowdown of $0.8\%$. Yet, since
\cref{thm:reduce}~(\markExistsReplace) does not apply to this benchmark and the version including
\cref{thm:affine}~(\markShiftReplace) is $0.7\%$ faster than the one only with
Optimisations~\ref{opt:and-transpose} and \ref{opt:prune}~(\markPruningAnd), this can be attributed
to noise and the whims of the compiler.

\subsection{RQ~\ref{rq:competitors:depth-first}: Comparison to Depth-first Implementations}

\subsubsection{Unlimited Memory}

\if\arxiv1%
  \begin{table}[p]
\else
  \begin{table}[t]
\fi
  \centering

  \caption{Total Running time of Adiar
    \if\arxiv1%
    (with Prop.~\ref{thm:naive}, \ref{thm:reduce}, and
    \ref{thm:affine} and Opt.~\ref{opt:and-transpose} and \ref{opt:prune})
    \fi
    and other implementations of BDDs to solve all benchmarks. The \# column indicates the number of
    instances that were solved by all BDD packages.}
  \label{tab:versus}

  \bgroup
  \def\arraystretch{1.1}
  \setlength\tabcolsep{5pt}

  \if\arxiv1%
    \begin{tabular}{l||c||rrrrrr}
  \else
    \begin{tabular}{l||c||rrrrrH}
  \fi
                          & \#  & Adiar      & BuDDy   & CAL      & CUDD     & LibBDD   & Sylvan
    \\ \hline \hline
    \if\arxiv1%
    \emph{Reachability}   & 149 & 3659.7     & 44.2    & 2989.7   & 118.4    & 797.9    & 62.4

    \\
    \emph{Deadlock}       & 149 & 223.8      & 11.3    & 12935.69 & 73.6     & 67.9     & 9.6
    \\ \hline
    \fi
    \emph{Next}           & 12  & 23628.1    & 1827.43 & TO       & 10021.34 & 2958.89  & --
    \\
    \emph{Prev}           & 10  & 76353.8    & 4252.37 & TO       & 6890.05  & 6815.761 & --
  \end{tabular}
  \egroup
\end{table}
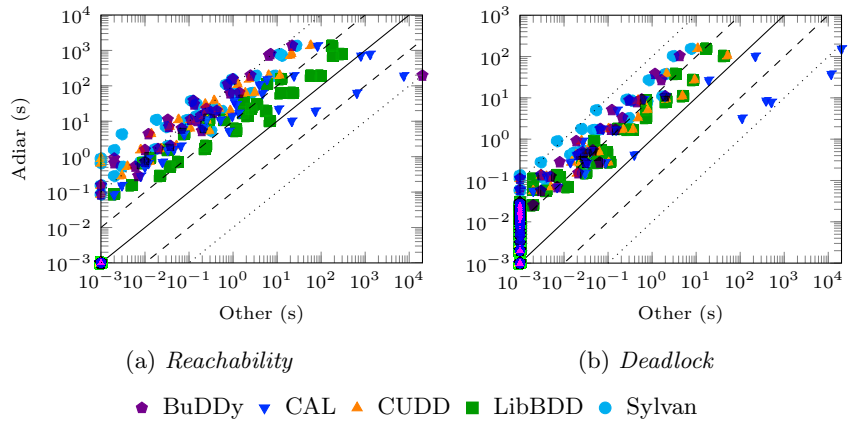
\begin{figure}[p]
  \centering

  \subfloat[\emph{Reachability}] {
    \label{fig:versus.mcnet:reachability}

    \begin{tikzpicture}
      \begin{axis}[%
        width=0.48\linewidth, height=0.4\linewidth,
        every tick label/.append style={font=\scriptsize},
        xlabel={\scriptsize Other (s)},
        xmin=0.001,
        xtick={0.001,0.01,0.1,1,10,100,1000,10000},
        xmax=20000,
        xmode=log,
        ylabel={\scriptsize Adiar (s)},
        ymin=0.001,
        ymax=10000,
        ytick={0.001,0.01,0.1,1,10,100,1000,10000},
        ymode=log,
        grid style={dashed,black!12},
        ]

        \addplot[domain=0.001:20000, samples=8, color=black, dotted]
        {0.01*x};
        \addplot[domain=0.001:20000, samples=8, color=black, dashed]
        {0.1*x};
        \addplot[domain=0.001:10000, samples=8, color=black]
        {x};
        \addplot[domain=0.001:1000, samples=8, color=black, dashed]
        {10*x};
        \addplot[domain=0.001:100, samples=8, color=black, dotted]
        {100*x};

        \begin{scope}[blend mode=soft light]
          \addplot+ [style=dots_buddy]
          table {./data/reachability.buddy.tex};

          \addplot+ [style=dots_cal]
          table {./data/reachability.cal.tex};

          \addplot+ [style=dots_cudd]
          table {./data/reachability.cudd.tex};

          \addplot+ [style=dots_libbdd]
          table {./data/reachability.libbdd.tex};

          \addplot+ [style=dots_sylvan]
          table {./data/reachability.sylvan.tex};
        \end{scope}
      \end{axis}
    \end{tikzpicture}
  }
  \subfloat[\emph{Deadlock}] {
    \label{fig:versus.mcnet:deadlock}

    \begin{tikzpicture}
      \begin{axis}[%
        width=0.48\linewidth, height=0.4\linewidth,
        every tick label/.append style={font=\scriptsize},
        xlabel={\scriptsize Other (s)},
        xmin=0.001,
        xtick={0.001,0.01,0.1,1,10,100,1000,10000,100000},
        xmax=20000,
        xmode=log,
        ymin=0.001,
        ymax=1000,
        ytick={0.001,0.01,0.1,1,10,100,1000},
        ymode=log,
        grid style={dashed,black!12},
        ]

        \addplot[domain=0.001:100000, samples=8, color=black, dotted]
        {0.01*x};
        \addplot[domain=0.001:100000, samples=8, color=black, dashed]
        {0.1*x};
        \addplot[domain=0.001:1000, samples=8, color=black]
        {x};
        \addplot[domain=0.001:100, samples=8, color=black, dashed]
        {10*x};
        \addplot[domain=0.001:10, samples=8, color=black, dotted]
        {100*x};

        \begin{scope}[blend mode=soft light]
          \addplot+ [style=dots_buddy]
          table {./data/deadlock.buddy.tex};

          \addplot+ [style=dots_cal]
          table {./data/deadlock.cal.tex};

          \addplot+ [style=dots_cudd]
          table {./data/deadlock.cudd.tex};

          \addplot+ [style=dots_libbdd]
          table {./data/deadlock.libbdd.tex};

          \addplot+ [style=dots_sylvan]
          table {./data/deadlock.sylvan.tex};
        \end{scope}
      \end{axis}
    \end{tikzpicture}
  }

  \smallskip

  \begin{tikzpicture}
    \begin{customlegend}[
      legend columns=-1,
      legend style={draw=none,column sep=1ex},
      legend entries={
        BuDDy,
        CAL,
        CUDD,
        LibBDD,
        Sylvan
      }
      ]
      \addlegendimage{style=dots_buddy}
      \addlegendimage{style=dots_cal}
      \addlegendimage{style=dots_cudd}
      \addlegendimage{style=dots_libbdd}
      \addlegendimage{style=dots_sylvan}
    \end{customlegend}
  \end{tikzpicture}

  \caption{Running time of Adiar
    on model checking tasks compared to other implementations. Timeouts are shown as markers at
    the top and the right.}
  \label{fig:versus.mcnet}
\end{figure}
\if\arxiv1%
  \begin{figure}[p]
\else
  \begin{figure}
\fi
  \centering

  \subfloat[\emph{Next}] {
    \begin{tikzpicture}
      \begin{axis}[%
        width=0.48\linewidth, height=0.4\linewidth,
        every tick label/.append style={font=\scriptsize},
        xlabel={\scriptsize Other (s)},
        xmin=10,
        xtick={10,100,1000,10000,100000},
        xmax=100000,
        xmode=log,
        ylabel={\scriptsize Adiar (s)},
        ymin=10,
        ymax=100000,
        ytick={1,10,100,1000,10000,100000},
        ymode=log,
        grid style={dashed,black!12},
        ]

        \addplot[domain=10:100000, samples=8, color=black, dotted]
        {0.01*x};
        \addplot[domain=10:100000, samples=8, color=black, dashed]
        {0.1*x};
        \addplot[domain=10:100000, samples=8, color=black]
        {x};
        \addplot[domain=10:10000, samples=8, color=black, dashed]
        {10*x};
        \addplot[domain=10:1000, samples=8, color=black, dotted]
        {100*x};

        \begin{scope}[blend mode=soft light]
          \addplot+ [style=dots_buddy]
          table {./data/next.buddy.tex};

          \addplot+ [style=dots_cal]
          table {./data/next.cal.tex};

          \addplot+ [style=dots_cudd]
          table {./data/next.cudd.tex};

          \addplot+ [style=dots_libbdd]
          table {./data/next.libbdd.tex};
        \end{scope}
      \end{axis}
    \end{tikzpicture}
  }
  \subfloat[\emph{Prev}] {
    \begin{tikzpicture}
      \begin{axis}[%
        width=0.48\linewidth, height=0.4\linewidth,
        every tick label/.append style={font=\scriptsize},
        xlabel={\scriptsize Other (s)},
        xmin=10,
        xtick={1,10,100,1000,10000,100000},
        xmax=100000,
        xmode=log,
        ymin=10,
        ymax=100000,
        ytick={1,10,100,1000,10000,100000},
        yticklabels={,,,,,},
        ymode=log,
        grid style={dashed,black!12},
        ]

        \addplot[domain=10:100000, samples=8, color=black, dotted]
        {0.01*x};
        \addplot[domain=10:100000, samples=8, color=black, dashed]
        {0.1*x};
        \addplot[domain=10:100000, samples=8, color=black]
        {x};
        \addplot[domain=10:10000, samples=8, color=black, dashed]
        {10*x};
        \addplot[domain=10:1000, samples=8, color=black, dotted]
        {100*x};

        \begin{scope}[blend mode=soft light]
          \addplot+ [style=dots_buddy]
          table {./data/prev.buddy.tex};

          \addplot+ [style=dots_cal]
          table {./data/prev.cal.tex};

          \addplot+ [style=dots_cudd]
          table {./data/prev.cudd.tex};

          \addplot+ [style=dots_libbdd]
          table {./data/prev.libbdd.tex};
        \end{scope}
      \end{axis}
    \end{tikzpicture}
  }

  \smallskip

  \begin{tikzpicture}
    \begin{customlegend}[
      legend columns=-1,
      legend style={draw=none,column sep=1ex},
      legend entries={
        BuDDy,
        CAL,
        CUDD,
        LibBDD
      }
      ]
      \addlegendimage{style=dots_buddy}
      \addlegendimage{style=dots_cal}
      \addlegendimage{style=dots_cudd}
      \addlegendimage{style=dots_libbdd}
    \end{customlegend}
  \end{tikzpicture}

  \caption{Running time of Adiar
    on a single relational product compared to other implementations. Timeouts are shown as markers
    at the top and the right.}
  \label{fig:versus.relprod}
\end{figure}
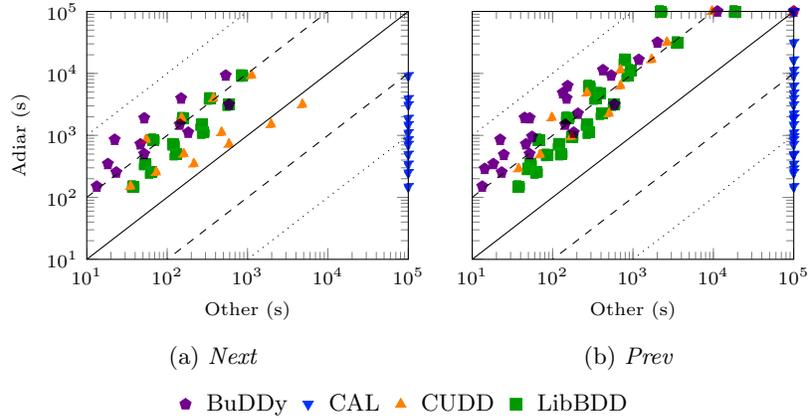

We have measured the running time for the same benchmarks with the depth-first BDD packages
BuDDy~2.4~\cite{Lind1999}, CUDD~3.0~\cite{Somenzi2015}, LibBDD~0.5.22~\cite{Benes2020}, and
Sylvan~1.8.1~\cite{Dijk2016}. Sylvan is not included for the \emph{Next} and \emph{Prev}
measurements since the manual BDD reconstruction results in a stack overflow\footnote{More
  precisely, its garbage collection algorithm starts out by creating a task for each BDD root. This
  is to mark all nodes that are still alive. Yet, if the input BDD is big enough then the number of
  BDD nodes at a single level may outgrow the worker queues in Lace~\cite{Dijk2014}.}.
\Cref{fig:versus.mcnet,fig:versus.relprod} show scatter plots of the running time of Adiar (with all
features from \cref{sec:relprod}) in comparison to the other implementations.

That the BDD size in \emph{Reachability} and \emph{Deadlock} are small is clearly evident in
\cref{tab:versus,fig:versus.mcnet}. Earlier, the disk-based algorithms of Adiar have proven to be
several orders of magnitude slower than the conventional depth-first algorithms when computing on
small BDDs \cite{Soelvsten2022:TACAS,Soelvsten2023:ATVA,Soelvsten2024:arXiv}. This is also evident
in \cref{fig:versus.mcnet} where the gap between Adiar and depth-first implementations becomes
smaller as the computation time, and with it the BDDs, grow larger. Furthermore, the gap in
\cref{fig:versus.mcnet:deadlock} is smaller than in \cref{fig:versus.mcnet:reachability} since
\emph{Deadlock} only includes computation on the larger BDD that represents the entire state space.

As the BDD sizes are larger in \emph{Next} and \emph{Prev}, the gap between Adiar and depth-first
implementations is expected to be smaller. Indeed, as \cref{tab:versus} and
\cref{fig:versus.relprod} show, the gap is only of a single order of magnitude or less.

\subsubsection{Limited Memory}

At time of writing, Adiar still mainly focuses on handling BDDs that are too large to fit into main
memory. That is, the results from our experiments shown in
\cref{fig:versus.mcnet,fig:versus.relprod} are heavily biased against Adiar. The benefit of doing
so is for the experiments to provide a worst-case comparison between Adiar and other
implementations.

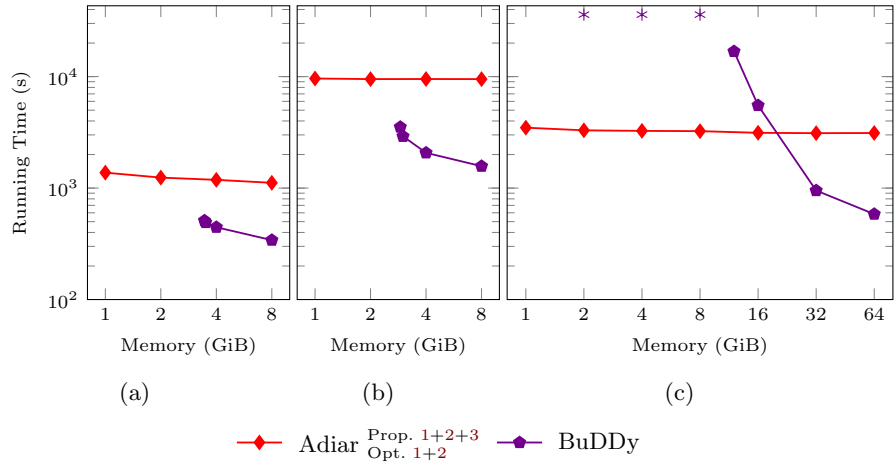
\begin{figure}[t]
  \centering

  \subfloat[] {
    \label{fig:memory:GPUForwardProgress}

    \begin{tikzpicture}
      \begin{axis}[%
        width=0.35\linewidth, height=0.45\linewidth,
        every tick label/.append style={font=\scriptsize},
        xlabel={\scriptsize Memory (GiB)},
        xmin=0.8,
        xtick={1,2,4,8},
        xticklabels={1,2,4,8},
        xmax=10,
        xmode=log,
        ylabel={\scriptsize Running Time (s)},
        ymin=100,
        ymax=43200,
        ytick={100,1000,10000,100000},
        ymode=log,
        grid style={dashed,black!12},
        ]


        \addplot+ [style=plot_adiar]
        table {./data/memory.gpufp_20_a.adiar.tex};

        \addplot+ [style=plot_buddy]
        table {./data/memory.gpufp_20_a.buddy.tex};
      \end{axis}
    \end{tikzpicture}
  }
  \hspace{-16pt}
  \subfloat[] {
    \label{fig:memory:SmartHome}

    \begin{tikzpicture}
      \begin{axis}[%
        width=0.35\linewidth, height=0.45\linewidth,
        every tick label/.append style={font=\scriptsize},
        xlabel={\scriptsize Memory (GiB)},
        xmin=0.8,
        xtick={1,2,4,8},
        xticklabels={1,2,4,8},
        xmax=10,
        xmode=log,
        ymin=100,
        ymax=43200,
        ytick={100,1000,10000,100000},
        yticklabels={,,,},
        ymode=log,
        grid style={dashed,black!12},
        ]


        \addplot+ [style=plot_adiar]
        table {./data/memory.smhome_16.adiar.tex};

        \addplot+ [style=plot_buddy]
        table {./data/memory.smhome_16.buddy.tex};
      \end{axis}
    \end{tikzpicture}
  }
  \hspace{-16pt}
  \subfloat[] {
    \label{fig:memory:ShieldPPPs}

    \begin{tikzpicture}
      \begin{axis}[%
        width=0.55\linewidth, height=0.45\linewidth,
        every tick label/.append style={font=\scriptsize},
        xlabel={\scriptsize Memory (GiB)},
        xmin=0.8,
        xtick={1,2,4,8,16,32,64},
        xticklabels={1,2,4,8,16,32,64},
        xmax=80,
        xmode=log,
        ymin=100,
        ymax=43200,
        ytick={100,1000,10000,100000},
        yticklabels={,,,},
        ymode=log,
        grid style={dashed,black!12},
        ]

        \addplot+ [style=x_buddy, forget plot] coordinates {
          (2, 36000)
          (4, 36000)
          (8, 36000)
        };

        \addplot+ [style=plot_adiar]
        table {./data/memory.shield_s_ppp_010_a.adiar.tex};

        \addplot+ [style=plot_buddy]
        table {./data/memory.shield_s_ppp_010_a.buddy.tex};
      \end{axis}
    \end{tikzpicture}
  }

  \smallskip

  \if\arxiv1%
  \begin{tikzpicture}
    \begin{customlegend}[
      legend columns=-1,
      legend style={draw=none,column sep=1ex},
      legend entries={
        {Adiar $^{\text{Prop.~\ref{thm:naive}+\ref{thm:reduce}+\ref{thm:affine}}}_{\text{Opt.~\ref{opt:and-transpose}+\ref{opt:prune}}}$},
        BuDDy
      }
      ]
      \addlegendimage{style=plot_adiar}
      \addlegendimage{style=plot_buddy}
    \end{customlegend}
  \end{tikzpicture}
  \else
  \begin{tikzpicture}
    \begin{customlegend}[
      legend columns=-1,
      legend style={draw=none,column sep=1ex},
      legend entries={
        Adiar,
        BuDDy
      }
      ]
      \addlegendimage{style=plot_adiar}
      \addlegendimage{style=plot_buddy}
    \end{customlegend}
  \end{tikzpicture}
  \fi

  \caption{Running time for \emph{Next} depending on available internal memory. Timeouts are
    marked as stars. The models are \texttt{GPUForwardProgress} 20a
    (\ref{fig:memory:GPUForwardProgress}), \texttt{SmartHome} 16 (\ref{fig:memory:SmartHome}), and
    \texttt{ShieldPPPs} 10a (\ref{fig:memory:ShieldPPPs}) with a state space BDD of magnitude $2^{25}$.}
  \label{fig:memory}
\end{figure}

To turn the tables, \cref{fig:memory} shows for each of the three models of the \emph{Next}
benchmark the running time of Adiar and BuDDy when given different amounts of internal memory. Each
of these data points were only measured once. In \texttt{GPUForwardProgress}~20a
(\cref{fig:memory:GPUForwardProgress}) and \texttt{SmartHome}~16 (\cref{fig:memory:SmartHome}),
BuDDy is unable to complete its computation with less than 3.45 and 2.9 GiB of memory, respectively.
As the memory increases beyond this point, its performance quickly increases. For the
\texttt{ShieldPPPs}~10a model (\cref{fig:memory:ShieldPPPs}), BuDDy needs more than 12~h to finish
its computation when given less than 12~GiB of memory. This is due to repeated need to do garbage
collection which in turn clears the computation cache and makes it repeat previous computations.
Adiar, on the other hand, does not use any memoisation; its priority queues are implicitly doing
double-duty as a computation cache \cite{Soelvsten2022:TACAS}. It neither suffers from garbage
collection: BDD nodes are stored in files on disk where the entire file is deleted when it is not
needed anymore \cite{Soelvsten2022:TACAS}. Yet, Adiar's use of I/O-efficient files and priority
queues allows it to only slow down by $23\%$ as its internal memory is decreased far below what
BuDDy needs.

\subsection{RQ~\ref{rq:competitors:breadth-first}: Comparison to CAL (Breadth-first Implementation)}

\Cref{tab:versus,fig:versus.mcnet,fig:versus.relprod} also include measurements of the breadth-first
BDD package, CAL~\cite{Sanghavi1996}. Similar to our previous results in \cite{Soelvsten2024:arXiv},
CAL's running time is on-par with the depth-first implementations for smaller instances but then
quickly deteriorates as the BDDs grow larger. This is due to its use of conventional depth-first
recursion for the smaller instances \cite{Soelvsten2023:ATVA}.

The overhead of its breadth-first algorithms, on the other hand, makes it much slower than Adiar for
the larger instances in \cref{fig:versus.mcnet:reachability,fig:versus.mcnet:deadlock}. This
overhead is especially apparent in the \emph{Next} and \emph{Prev}. Here, CAL times out after 12~h
for all instances.

\section{Conclusions and Future Work} \label{sec:conclusion}

In this work, we have succesfully designed an I/O-efficient relational product in Adiar which
enables BDD-based symbolic model checking beyond the limits of the machine's memory. These
algorithms, as they exist today, are much faster at processing large BDDs than a conventional BDD
package that either needs to repeatedly run garbage collection to stay within its memory limits or
that needs to offload its BDDs to the disk by means of swap memory.
In fact, Adiar's running time is virtually independent of its memory limits.

\subsection*{Optimisations for Relational Product}

Towards designing this I/O-efficient relational product, we have identified multiple optimisations
particular to the design of Adiar's algorithms. Based on our results in \cref{sec:experiments}, we
recommend the following with respect to the optimisations we have proposed in \cref{sec:relprod}.

\begin{recommendation}
  \Cref{thm:naive,thm:reduce} and Optimisations \ref{opt:and-transpose} and \ref{opt:prune} should
  be included in an I/O-efficient relational product. A safer alternative to \cref{thm:affine}
  may be worth the additional implementation effort.
\end{recommendation}

The gap in running time between Adiar and depth-first implementations is for larger instances about
one order of magnitude. This gap is about twice the size than our previous results in
\cite{Soelvsten2022:TACAS,Soelvsten2023:NFM,Soelvsten2023:ATVA,Soelvsten2024:arXiv}: in those
benchmarks, Adiar is only up to a factor 4 slower than the conventional approach. In
\cite{Dijk2015}, a combined \texttt{AndExists} BDD operation roughly halves the running time for
depth-first implementations. The optimisations in \cref{sec:and-exists} aim to also create a
combined \texttt{AndExists} within the time-forward processing paradigm of Adiar. The gap in
performance indicates more ideas are needed.

\begin{recommendation}
  Future work on an I/O-efficient relational product should further integrate the \texttt{Exists}
  within the \texttt{And}. To this end, it may be worth investigating alternatives to
  Optimisation~\ref{opt:prune}; for example, partial quantification in \cite{Soelvsten2024:arXiv}
  may prove performant in the context of symbolic model checking with BDDs.
\end{recommendation}

\noindent Additional measurements have indicated that for larger instances, the \texttt{And} (even
without the optimisations in \cref{sec:and-exists}) is responsible for less than $1/10$th of the
total computation time.

\begin{recommendation}
  Future work on an I/O-efficient relational product should especially focus on improving the
  I/O-efficient \texttt{Exists} operation's performance.
\end{recommendation}

\subsection*{Small-scale BDD Computation}

The BDD library CAL~\cite{Sanghavi1996} (based on \cite{Ochi1993,Ashar1994}) is to the best of our
knowledge the only other implementation aiming at efficiently computing on BDDs stored on the disk.
In practice, it switches from the conventional depth-first to the breadth-first algorithms described
in \cite{Sanghavi1996} when the size of the BDDs exceed $2^{19}$ BDD nodes
\cite{Soelvsten2023:ATVA}. Yet, these breadth-first algorithms are in practice one or more orders of
magnitude slower than Adiar's algorithms. This makes Adiar vastly outperform CAL at solving our
benchmarks as the BDDs involved have grown large enough.

But, similar to our previous results in
\cite{Soelvsten2022:TACAS,Soelvsten2023:ATVA,Soelvsten2024:arXiv}, the gap in performance between
Adiar and depth-first implementations is still several orders of magnitude for smaller instances. As
the BDDs in model checking tasks start out small and only grow slowly, this overhead bars the
current version of Adiar from being applicable for model checking in practice.

\begin{recommendation}
  Similar to CAL, one should combine Adiar's time-forward processing algorithms with the
  conventional depth-first approach. Both \cite{Soelvsten2023:ATVA} and \cite{Soelvsten2024:SPIN}
  have paved the way for getting Adiar's algorithms to work in tandem with a unique node table.
\end{recommendation}

\noindent Yet, doing such a vast engineering task is outside the scope of this paper. We leave the
task of combining the strengths of depth-first recursion and time-forward processing as future work.


\subsection*{Generic Variable Substitution}

We have in this work only focused on monotone variable substitution. This still leaves a variable
substitution that is I/O-efficient for the general case as an open problem, i.e.\ an algorithm
capable of handling substitutions that change the order of the BDD's levels. If it is possible to
design such an algorithm which is efficient with respect to time, space, and I/Os then it can be
used as the backbone for novel external memory variable reordering algorithms.

\section*{Acknowledgements}

Thanks to the Centre for Scientific Computing, Aarhus, (\href{https://www.cscaa.dk/}{www.cscaa.dk/})
for running our benchmarks on their Grendel cluster. Furthermore, thanks to Nils Husung for
previously having added LibBDD to the BDD Benchmarking
Suite~\cite{Soelvsten2022:TACAS,Soelvsten2021:Zenodo} as part of \cite{Husung2024}; this turned out
to be vital for creating inputs for the \emph{Next} and \emph{Prev} benchmarks.

\clearpage
\section*{Data Availability Statement}

Our experiments are created based on the BDD Benchmarking
Suite~\cite{Soelvsten2022:TACAS,Soelvsten2021:Zenodo}. The obtained data and its analysis are
available in our accompanying artifact \cite{Soelvsten2025:Zenodo}. This artifact also includes the
inputs, pre-compiled binaries, and scripts to recreate our experiments. Finally, it also provides
the source code to read, change, and recompile said binaries.


%
%
%

\DeclareRobustCommand{\VAN}[3]{#2} 

\bibliographystyle{splncs04}
\bibliography{references}

\iftrue

  \clearpage
  \appendix
  \section{SCC Decomposition}

  We also implemented the following third foundational model checking operation for our experiments in
  \cref{sec:experiments}.
  \begin{itemize}
  \item\emph{SCC Decomposition}:

    The reachable states are decomposed into their strongly connected components (SCCs) via the
    \textsc{Chain} algorithm~\cite{Larsen2023}.
    This uses both \texttt{bdd\_relnext} and \texttt{bdd\_relprev} up to a polynomial number of times
    with respect to the model and its state space.
    As per \cite{Jakobsen2024}, each deadlock state identified during the \emph{Deadlock} stage is an
    SCC which can be skipped.
  \end{itemize}
  The 75 model instances were in particular the ones where Adiar seemed able to consecutively solve
  \emph{Reachability}, \emph{Deadlock}, and \emph{SCC Decomposition} within 2 days. Likewise, the
  results in measurements in \cref{tab:versus,fig:versus.mcnet} were run with a timeout of 48~h.

  This benchmark was left out of the above experiments, since they do not add new knowledge. In
  particular, the BDD size stayed too small to be within Adiar's current scope. If anything, the
  results shown below further cements the need for incorporating the conventional depth-first
  algorithms with Adiar's I/O-efficient time-forward processing.

  \subsection{RQ~\ref{rq:optimisations}: Effect of the Optimisations}

  \Cref{tab:versions:scc,fig:versions.mcnet:scc} shows the effect of each optimisation on this
  particular benchmark. Similar to \emph{Reachability} and \emph{Deadlock}, the \emph{SCC
    Decomposition} is positively affected by the optimisations due to the small BDD size.

  Optimisation~\ref{opt:and-transpose}~(\markSkipTranspose) improves the total running time for
  \emph{SCC Decomposition} with $20.7\%$. Optimisation~\ref{opt:prune}~(\markPruningAnd) further
  improves the running time by $14.0\%$ and \Cref{thm:reduce}~(\markExistsReplace) by $4.8\%$.
  Finally, \cref{thm:affine}~(\markShiftReplace) improves the total running time by $1.5\%$

  \begin{table}[b]
  \centering

  \caption{Total running time (seconds) of each version of Adiar on SCC Decomposition. The \# column
    indicates the number of instances that were solved by all five versions.}
  \label{tab:versions:scc}

  \bgroup
  \def\arraystretch{1.1}
  \setlength\tabcolsep{5pt}

  \begin{tabular}{l||c||rrrrr}
                          & \#  & {---$^{\text{Prop.~\ref{thm:naive}}}$}
                                             & {\markSkipTranspose$^{\text{Prop.~\ref{thm:naive}}}_{\text{Opt.~\ref{opt:and-transpose}}}$}
                                                          & {\markPruningAnd$^{\text{Prop.~\ref{thm:naive}}}_{\text{Opt.~\ref{opt:and-transpose}+\ref{opt:prune}}}$}
                                                                       & {\markExistsReplace$^{\text{Prop.~\ref{thm:naive}+\ref{thm:reduce}}}_{\text{Opt.~\ref{opt:and-transpose}+\ref{opt:prune}}}$}
                                                                                    & {\markShiftReplace$^{\text{Prop.~\ref{thm:naive}+\ref{thm:reduce}+\ref{thm:affine}}}_{\text{Opt.~\ref{opt:and-transpose}+\ref{opt:prune}}}$}
    \\ \hline \hline
    \emph{SCC}            & 144 & 734040.5   & 562074.5   & 483553.0   & 460388.5   & 453712.4
  \end{tabular}
  \egroup
\end{table}
  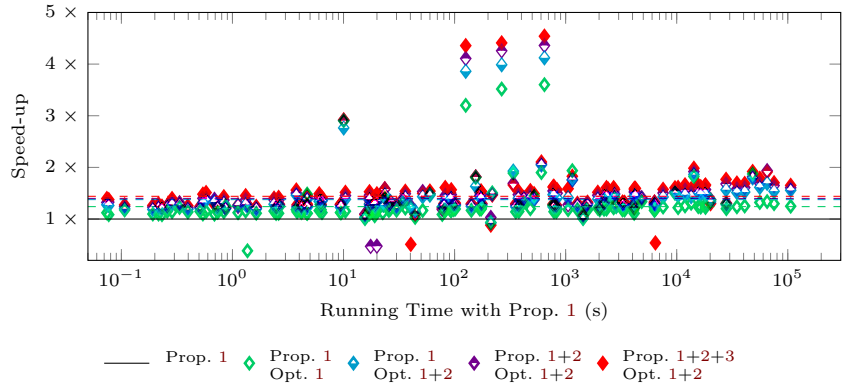
\begin{figure}[t]
  \centering

  \begin{tikzpicture}
    \begin{axis}[%
      width=0.95\linewidth, height=0.4\linewidth,
      every tick label/.append style={font=\scriptsize},
      xlabel={\scriptsize Running Time with Prop.~\ref{thm:naive} (s)},
      xmin=0.05,
      xtick={0.01,0.1,1,10,100,1000,10000,100000,1000000},
      xmax=300000,
      xmode=log,
      ylabel={\scriptsize Speed-up},
      ymin=0.2,
      ymax=5.0,
      ytick={1,2,3,4,5},
      yticklabels={
        1 $\times$,
        2 $\times$,
        3 $\times$,
        4 $\times$,
        5 $\times$,
      },
      grid style={dashed,black!12},
      ]

      \addplot[domain=0.05:300000, samples=8, color=black]
      {1};

      \begin{scope}[blend mode=soft light]
        \addplot+ [style=dots_adiar-skip_transpose]
        table {./data/scc.adiar-skip_transpose.tex};

        \addplot+ [style=dots_adiar-pruning_and]
        table {./data/scc.adiar-pruning_and.tex};

        \addplot+ [style=dots_adiar-exists_replace]
        table {./data/scc.adiar-exists_replace.tex};

        \addplot+ [style=dots_adiar-shift_replace]
        table {./data/scc.adiar-shift_replace.tex};

        \addplot[domain=0.005:300000, samples=8, style=dashed, opacity=0.5, color=adiar-skip_transpose]
        {1.240490326254832};
        \addplot[domain=0.005:300000, samples=8, style=dashed, opacity=0.5, color=adiar-pruning_and]
        {1.3761794247188053};
        \addplot[domain=0.005:300000, samples=8, style=dashed, opacity=0.5, color=adiar-exists_replace]
        {1.3994938665975531};
        \addplot[domain=0.005:300000, samples=8, style=dashed, opacity=0.5, color=red]
        {1.440069060711158};
      \end{scope}
    \end{axis}
  \end{tikzpicture}

  \smallskip

  \begin{tikzpicture}
    \begin{customlegend}[
      legend columns=-1,
      legend style={draw=none,column sep=1ex},
      legend entries={
        {$^{\text{Prop.~\ref{thm:naive}}}$},
        {$^{\text{Prop.~\ref{thm:naive}}}_{\text{Opt.~\ref{opt:and-transpose}}}$},
        {$^{\text{Prop.~\ref{thm:naive}}}_{\text{Opt.~\ref{opt:and-transpose}+\ref{opt:prune}}}$},
        {$^{\text{Prop.~\ref{thm:naive}+\ref{thm:reduce}}}_{\text{Opt.~\ref{opt:and-transpose}+\ref{opt:prune}}}$},
        {$^{\text{Prop.~\ref{thm:naive}+\ref{thm:reduce}+\ref{thm:affine}}}_{\text{Opt.~\ref{opt:and-transpose}+\ref{opt:prune}}}$},
      }
      ]
      \addlegendimage{style=dots_adiar-naive}
      \addlegendimage{style=dots_adiar-skip_transpose}
      \addlegendimage{style=dots_adiar-pruning_and}
      \addlegendimage{style=dots_adiar-exists_replace}
      \addlegendimage{style=dots_adiar-shift_replace}
    \end{customlegend}
  \end{tikzpicture}

  \caption{Impact of optimisations on SCC Decomposition running time. Averages are drawn as dashed
    lines.}
  \label{fig:versions.mcnet:scc}
\end{figure}

  \subsection{RQ~\ref{rq:competitors:depth-first}: Comparison to Depth-first Implementations}

  \Cref{tab:versus:scc,fig:versus.mcnet:scc} shows the running time of Adiar and the depth-first BDD
  packages described in \cref{sec:experiments}. Here, the gap between Adiar and the depth-first BDD
  packages clearly show that the size of the BDDs stay small. This is due to the fact that the Chain
  algorithm in \cite{Larsen2023} repeatedly explores the (remaining) set of states from single pivot
  states.

  \if\arxiv1%
  \begin{table}[b!]
\else
  \begin{table}[t]
\fi
  \centering

  \caption{Total Running time of Adiar (with Prop.~\ref{thm:naive}, \ref{thm:reduce}, and
    \ref{thm:affine} and Opt.~\ref{opt:and-transpose} and \ref{opt:prune}) and other implementations
    of BDDs for SCC decomposition. The \# column indicates the number of instances that were solved
    by all BDD packages.}
  \label{tab:versus:scc}

  \bgroup
  \def\arraystretch{1.1}
  \setlength\tabcolsep{5pt}

  \begin{tabular}{l||c||rrrrrr}
                          & \#  & Adiar      & BuDDy   & CAL      & CUDD     & LibBDD   & Sylvan
    \\ \hline \hline
    \emph{SCC}            & 147 & 567188.5   & 680.6   & 22679.93 & 1840.0   & 10201.9  & 862.2
    \\ \hline
  \end{tabular}
  \egroup
\end{table}
  \begin{figure}[b!]
  \centering

  \begin{tikzpicture}
    \begin{axis}[%
      width=0.95\linewidth, height=0.4\linewidth,
      every tick label/.append style={font=\scriptsize},
      xlabel={\scriptsize Other (s)},
      xmin=0.001,
      xtick={0.001,0.01,0.1,1,10,100,1000,10000,100000},
      xmax=200000,
      xmode=log,
      ylabel={\scriptsize Adiar (s)},
      ymin=0.001,
      ymax=200000,
      ytick={0.001,0.01,0.1,1,10,100,1000,10000,100000},
      ymode=log,
      grid style={dashed,black!12},
      ]

      \addplot[domain=0.001:200000, samples=8, color=black, dotted]
      {0.01*x};
      \addplot[domain=0.001:200000, samples=8, color=black, dashed]
      {0.1*x};
      \addplot[domain=0.001:200000, samples=8, color=black]
      {x};
      \addplot[domain=0.001:20000, samples=8, color=black, dashed]
      {10*x};
      \addplot[domain=0.001:2000, samples=8, color=black, dotted]
      {100*x};

      \begin{scope}[blend mode=soft light]
        \addplot+ [style=dots_buddy]
        table {./data/scc.buddy.tex};

        \addplot+ [style=dots_cal]
        table {./data/scc.cal.tex};

        \addplot+ [style=dots_cudd]
        table {./data/scc.cudd.tex};

        \addplot+ [style=dots_libbdd]
        table {./data/scc.libbdd.tex};

        \addplot+ [style=dots_sylvan]
        table {./data/scc.sylvan.tex};
      \end{scope}
    \end{axis}
  \end{tikzpicture}

  \begin{tikzpicture}
    \begin{customlegend}[
      legend columns=-1,
      legend style={draw=none,column sep=1ex},
      legend entries={
        BuDDy,
        CAL,
        CUDD,
        LibBDD,
        Sylvan
      }
      ]
      \addlegendimage{style=dots_buddy}
      \addlegendimage{style=dots_cal}
      \addlegendimage{style=dots_cudd}
      \addlegendimage{style=dots_libbdd}
      \addlegendimage{style=dots_sylvan}
    \end{customlegend}
  \end{tikzpicture}

  \caption{Running time of Adiar
    on SCC decomposition compared to other implementations. Timeouts are shown as markers at the top
    and the right.}
  \label{fig:versus.mcnet:scc}
\end{figure}
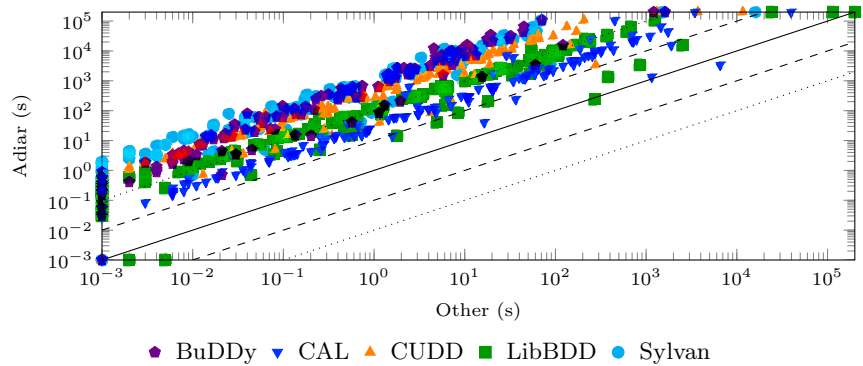

  \subsection{RQ~\ref{rq:competitors:breadth-first}: Comparison to CAL (Breadth-first Implementation)}

  \Cref{tab:versus:scc,fig:versus.mcnet:scc} also shows the running time of CAL for solving the SCC
  decomposition taks. Whereas CAL becomes slower than Adiar for the largest instances in
  \cref{fig:versus.mcnet:reachability,fig:versus.mcnet:deadlock}, the same is not evident in
  \cref{fig:versus.mcnet:scc}. This is further testament to the small size of the BDDs in this
  benchmark.

\fi

\end{document}